%% file: main.tex
\documentclass{llncs}
\usepackage{geometry} 

\usepackage{amsthm}
\usepackage{amsfonts}
\usepackage{bbm}
\usepackage{physics}
\usepackage{comment}
\usepackage{amssymb}
\usepackage{breakcites}
\title{Franchised Quantum Money}
\author{Bhaskar Roberts\inst{1} and Mark Zhandry\inst{2,3}}
\date{}
\institute{UC Berkeley \and Princeton University \and NTT Research}

\input{macros}

\begin{document}
\maketitle

\begin{abstract}
The construction of public key quantum money based on standard cryptographic assumptions is a longstanding open question. Here we introduce franchised quantum money, an alternative form of quantum money that is easier to construct. Franchised quantum money retains the features of a useful quantum money scheme, namely unforgeability and local verification: anyone can verify banknotes without communicating with the bank. In franchised quantum money, every user gets a unique secret verification key, and the scheme is secure against counterfeiting and sabotage, a new security notion that appears in the franchised model. Finally, we construct franchised quantum money and prove security assuming one-way functions. 
\end{abstract}

\input{a_intro}

\input{b_prelim}
\input{c_Definition}
\input{d_Simple_Construction}
\input{f_Full_Construction}
\input{e_Dist_Game}

\section*{Acknowledgements}

This work is supported in part by NSF. Any opinions, findings and conclusions or recommendations expressed in this material are those of the author(s) and do not necessarily reflect the views of NSF.

This work is also supported by MURI Grant FA9550-18-1-0161 and ONR award N00014-17-1-3025.

We thank Zeph Landau, Umesh Vazirani, and the Princeton Writing Center for helpful feedback on various drafts of this paper.

\bibliographystyle{alpha}
\bibliography{references}
\end{document}

%% file: macros.tex
\newcommand{\sigkeygen}{{\sf SigKeyGen}}
\newcommand{\sign}{{\sf Sign}}
\newcommand{\sigver}{{\sf SigVer}}
\newcommand{\setup}{{\textsf{Setup}}}
\newcommand{\franchise}{{\textsf{Franchise}}}
\newcommand{\mint}{{\textsf{Mint}}}
\newcommand{\ver}{{\textsf{Ver}}}
\newcommand{\enckeygen}{\textsf{EncKeyGen}}
\newcommand{\enc}{{\textsf{Enc}}}
\newcommand{\dec}{{\textsf{Dec}}}
\newcommand{\negl}{{\sf negl}}
\newcommand{\poly}{{\sf poly}}

\newcommand{\bv}{{\mathbf{v}}}
\newcommand{\bw}{{\mathbf{w}}}
\newcommand{\bx}{{\mathbf{x}}}
\newcommand{\by}{{\mathbf{y}}}

\newcommand{\Nbb}{\mathbb{N}}

\newcommand{\Hs}{{\mathcal{H}}}
\newcommand{\Z}{{\mathbb{Z}}}

%% file: a_intro.tex
\section{Introduction}\label{sec:intro}

The application of quantum information to unforgeable currency was first envisioned by Wiesner~\cite{Wiesner83}, and these early ideas laid the foundation for the field of quantum cryptography. However, Wiesner's scheme for quantum money has a major drawback: verifying that a banknote is valid requires a classical description of the state, so the banknote must be sent back to the bank for verification.

The key properties that make cash (paper bills) useful are that anyone can verify banknotes \textit{locally}, without communicating with the bank, and the banknotes are hard to counterfeit. In a classical world, digital currency cannot hope to achieve these properties because any classical bitstring can be duplicated. In a quantum world, we have hope for uncounterfeitable money because of the no-cloning theorem.

Recent works~\cite{CCC:Aaronson09,ITCS:FGHLS12,AC12,Zha19} have sought a \emph{public} test to verify banknotes. A scheme with such a test is called  public key quantum money (or PKQM). Unfortunately, a convincing construction of public key quantum money 
has been notoriously elusive. Most proposals have been based on new ad hoc complexity assumptions, and in many cases those assumptions were broken~\cite{ITCS:FGHLS12,PKC:PenFauPer15,Aaronson16}. Recently, Zhandry~\cite{Zha19} showed that the 
\cite{AC12} scheme can be instantiated using recent indistinguishability obfuscators. However, the quantum security of such obfuscators is currently unclear. Zhandry also proposed a new quantum money scheme in \cite{Zha19}, but the security of his scheme was also called into question~\cite{Roberts19}.

\paragraph{Franchised Quantum Money:} In this work, we introduce franchised quantum money (FQM), which is useful as a currency system, easier to construct than public key quantum money, and potentially a stepping stone to PKQM.
In franchised quantum money, every user receives a unique secret verification key. With their key, a user can verify banknotes locally, but they cannot create counterfeit money that would fool another user. 
Our main result is to show how to realize franchised quantum money under essentially minimal assumptions, namely one-way functions.

Franchised quantum money is a secret key scheme that approximates the functionality of a public key scheme.  
In particular, franchised quantum money achieves local verification\footnote{\cite{BS20} also propose a quantum money scheme that tries to approximate the functionality of PKQM. However, their scheme does not achieve local verification: their banknotes must be periodically sent back to the bank for verification. Furthermore, the way they define security is hard to justify.}.

The franchised verification model is broadly useful for approximating the security guarantees of public key verification. Building off of an earlier, unpublished version of this paper, \cite{KNY21} proposed a franchised verification model for quantum lightning, and combined with a lattice assumption that we also proposed, they constructed a scheme for secure software leasing. 
\\

The central feature of franchised quantum money is that each user has a unique secret key. Furthermore, we only require that an adversary cannot trick a \textit{different} user into accepting a counterfeit banknote. 

The difficulty with PKQM is that if the adversary knows the verification key, they know what properties of the state will be tested during verification. It is hard to design a verification procedure that reveals just enough information to verify banknotes, without giving enough information to create fake banknotes that fool the verifier.

Franchised quantum money does not have this issue. The adversary does not know any other user's key, so they don't know what properties the other user will test during verification. Therefore it is hard for the adversary to trick the other user into accepting a counterfeit banknote.

\subsection{Technical Details} \label{technical_details}

\subsubsection*{Definition of Franchised Quantum Money:}

In franchised quantum money, there is a trusted party, called the bank, that administers the currency system by generating verification keys and banknotes. A banknote is valid if it was generated by the bank. 

The other participants in the system are untrusted users, who send and receive banknotes among each other. Each user can request a unique secret verification key from the bank. The key allows the user to verify any banknote they receive, and valid banknotes are accepted by verification with overwhelming probability. 

Some users (the adversaries) are malicious and try to trick other users into accepting invalid banknotes. However it's hard for an adversary to create invalid banknotes that another user would accept.

\subsubsection*{Security:}
In order to be considered secure, a franchised quantum money scheme must be secure against both counterfeiting and sabotage.

\paragraph{Security against counterfeiting:}
We say that the scheme is \textit{secure against counterfeiting} if it is hard for an adversary with $m$ valid banknotes to get any other users to accept $m+1$ banknotes. The key difference from public key quantum money lies in the word \textit{other}. We don't care if the adversary can produce $m+1$ banknotes that they themself would accept. 

In fact in our construction, it's easy for the adversary to ``trick themself'' into accepting invalid banknotes, because if they know what key will be used in verification, they can create invalid banknotes that will be accepted. 
However, a different user with a key that is unknown to the adversary will recognize these banknotes as invalid.

\paragraph{Security against sabotage:} 
Because each user has a different key, there is a second kind of security we need to consider. We don't want one user to accept an invalid banknote that another user would reject.

We call this attack sabotage:\footnote{We borrow this name from \cite{BS20}.} the adversary takes a valid banknote and modifies it. Then they give it to one user, who accepts it even though the banknote is invalid. But when the first user tries to spend the banknote with a second user, the second user rejects the banknote. 

How could sabotage be possible if the scheme is secure against counterfeiting? The adversary does not need to spend more banknotes than they received in order to succeed at sabotage.

A scheme is \textit{secure against sabotage} if the adversary cannot produce a banknote that one other user accepts but which a second other user rejects. 

\begin{remark}We note that sabotage attacks are also a potential concern for public key quantum money schemes. Even though all users run the same verification procedure, technically two successive runs of the procedure may not output the same result. However, this problem can always be avoided by implementing verification as a projective measurement. 

Furthermore, in practice, decoherence between runs may cause successive runs to behave differently. In this case too, sabotage attacks may be relevant.

To the best of our knowledge, this is the first work to point out these potential problems.
\end{remark}

If an FQM scheme is secure against counterfeiting and sabotage, then it is practically useful as currency. This is because users can trust that any banknote they accept will be accepted by all other users, and the money supply will not increase unless the bank produces more banknotes. Therefore, these banknotes can hold monetary value. Quantum money does not need to be public key in order to be useful as a currency system.

\subsubsection*{Construction from Hidden Subspaces: }
Our construction of FQM is based on \cite{AC12}'s proposal for PKQM from black-box subspace oracles. Below is a simplified version of our construction. A less-simplified version is given in section \ref{sec:simple_construction}, and the full version is given in section \ref{sec:construction}.

\paragraph{Banknote:} The banknote is an $n$-qubit quantum state. We can think of its computational basis states as vectors in $\Z_2^n$. The banknote $\ket{A}$ is a superposition over some random subspace $A \leq \Z_2^n$ such that $dim(A) = dim(A^\perp) = n/2$. We call this state a subspace state.
\[\ket{A} = \frac{1}{\sqrt{|A|}} \sum_{\bx \in A}\ket{\bx}\]

\paragraph{Verification key:} For a given banknote $\ket{A}$, each verification key is a pair of random subspaces ($V,W$). $V \leq A$ and $W \leq A^\perp$, and the dimension of $V$ and $W$ is $t := \Theta(\sqrt{n})$. Each verifier gets an independently random ($V,W$).

\paragraph{Verification:} To verify a banknote, the verifier performs two tests, one in the computational basis, and one in the Fourier basis. 

First we test that the classical basis states of $\ket{A}$ are in $W^\perp$.

Then we take the quantum Fourier transform of the banknote. If the banknote is valid, the resulting state, $\widetilde{\ket{A}}$, is a superposition over $A^\perp$ (\cite{AC12}):
\[\widetilde{\ket{A}} = \ket{A^\perp} = \frac{1}{\sqrt{|A^\perp|}} \sum_{\by \in A^\perp}\ket{\by}\]

Next, in the Fourier basis, we test that the vectors in $\widetilde{\ket{A}}$'s superposition are in $V^\perp$. Finally we take the inverse quantum Fourier transform, and return the resulting state. We accept the banknote if both tests passed. If the banknote was valid, the final state is the same as the initial one.

\subsubsection*{Discussion:} 
A verifier will accept any subspace state $\ket{B}$ where $V \leq B \leq W^\perp$. Note that the adversary can easily construct a $\ket{B}$ based on their key ($V, W$) that they themself would accept.

However, an adversary cannot trick other users into accepting an invalid banknote. With probability overwhelming in $n$, the other user's ($V,W$) include dimensions of $A$ and $A^\perp$, respectively, that are unknown to the adversary. Any banknote the adversary tries to produce, other than an honest banknote, will almost certainly get ``caught'' by these other dimensions and rejected.

\paragraph{Multiple banknotes.} In the simplified construction above, one verification key ($V,W$) cannot verify multiple banknotes. Each banknote uses a different subspace $A$, and $(V, W)$ depend on the choice of $A$. 

However in the full construction, one verification key needs to verify every banknote the user receives. To achieve this, we assume the existence of one-way functions, which implies CPA-secure encryption. First, ($V,W$) are encrypted and appended to the banknote as a classical ciphertext. Then the decryption key serves as the verification key -- the verifier decrypts the ciphertext to get ($V, W$), which they use to verify the banknote.

It is straightforward to see that \emph{some} computational assumptions are necessary for franchised quantum money, since given an unlimited number of banknotes, the bank's master secret key is information-theoretically determined. So our construction of franchised quantum money uses essentially minimal assumptions.

\paragraph{Franchised vs. Obfuscated Verification:}
The franchised verification model allows us to avoid using obfuscation when constructing quantum money, and the model may be useful beyond quantum money as a way to avoid obfuscation.

\cite{AC12, Zha19}'s construction of PKQM relies on strong forms of obfuscation, such as post-quantum-secure iO, for which we we have no convincing construction. The PKQM construction is like our FQM construction, except every verifier uses $V = A$ and $W = A^\perp$. We call this \textit{full} verification, in contrast to franchised verification. Additionally, the oracles checking membership in $A$ and $A^\perp$ are obfuscated so the adversary can't learn $A$.

In the franchised model, there is no need for obfuscation. The adversary only gets query access to the verifier, and they do not know the other users' verification keys. It is therefore feasible to construct FQM from assumptions weaker than obfuscation.

Finally, the franchised verifiers enjoy essentially the same security as full verifiers. We will show that the adversary cannot distinguish whether they're interacting with a full verifier or a franchised verifier, so our FQM construction inherits the security guarantees of the PKQM construction.

\subsubsection*{Colluding adversaries:}
As we defined FQM above, each user receives one verification key. But in the real world, it's possible that multiple adversaries collude: they pool their verification keys to gain more counterfeiting or sabotage power.

In our construction, each key gives a small number of dimensions of $A$ and $A^\perp$. If the adversary has unlimited verification keys, then they can learn all of $A$ and $A^\perp$ and produce as many copies of $\ket{A}$ as they want. So we will impose a collusion bound: no more than $C = \frac{n}{4t}$ adversaries can work together. This means no adversary learns more than $n/4$ dimensions of $A$ (or $A^\perp$). With this collusion bound, the scheme is secure.\par

Although our scheme needs large banknotes to handle a large collusion bound, this may be reasonable in any scenario where the number of users is small -- for example, in markets for certain financial securities, event tickets, etc. \footnote{We thank an anonymous reviewer for suggesting these applications.}

Additionally, collusion bounds are commonplace in cryptography, for example in traitor tracing. Our construction is analogous to the early days of traitor tracing, where the initial schemes~\cite{C:ChoFiaNao94} had ciphertexts with size linear in the collusion bound, and the main goal became to shrink the ciphertext size. Eventually, \cite{STOC:GoyKopWat18} essentially removed the collusion bound, giving a construction that is secure against exponentially many colluding adversaries, as a function of the ciphertext size.

Finally, we expect that any FQM scheme will require a collusion bound of some kind or else it would likely yield PKQM. See section \ref{discussion_and_open_problems} for more detail.
\subsection{Next Steps} 
\label{discussion_and_open_problems}

\paragraph{Increase the collusion bound:}
The main open problem is to increase the collusion bound, while maintaining small banknotes and verification keys. In our construction of FQM, the size of the banknotes ($n$) grows faster than the collusion bound ($C = \Theta(\sqrt{n})$). A reasonable next step is to construct a scheme whose banknote size grows slower than the collusion bound. 

Here are two possible approaches: first, we might use LWE or similar assumptions to add noise to the verification keys. Given many noisy keys, an adversary would hopefully be unable to learn the secret information needed for counterfeiting. LWE has been used in traitor tracing~\cite{STOC:GoyKopWat18} to increase the collusion bound while achieving short ciphertexts and secret keys (which are analogous to banknotes and verification keys).

Second, we can use combinatorial techniques, such as those used for traitor tracing in \cite{CCS:BonNao08}. \cite{CCS:BonNao08}'s techniques have resulted in optimally short ciphertexts and might be used to achieve short banknotes. 
However, combinatorial techniques in traitor tracing usually come at the cost of much larger secret keys, and we might expect something similar for franchised quantum money. 

\paragraph{Work up to public key quantum money:}
Franchised quantum money is a potential stepping stone to PKQM. Intuitively, the larger the collusion bound, the more the scheme behaves like PKQM, and we expect that PKQM can be easily constructed from an FQM construction that has unbounded collusion.

Hypothetically, how would we prove security for an FQM scheme with unbounded collusion? The reduction would have to generate the adversary's verification keys, and somehow use the adversary's forgery for honest keys to break some underlying hard problem. But if the reduction could generate new verification keys for itself, then the construction might also be able to generate these new keys. If this were the case, we would easily get a public key quantum money scheme: to verify a banknote, generate a new verification key for yourself, and use that key.

\paragraph{Franchised semi-quantum money:} We can make the mint in our scheme entirely classical, similar to the semi-quantum money scheme of~\cite{RadSat19}, which is a secret key scheme. This follows from the fact that anyone can create new (un-signed) banknotes. To create and send a new banknote to a recipient, the recipient will generate a new un-signed banknote $|\$\rangle$ with serial number $\by$ on its own. It will then send $\by$ to the mint, who will sign $\by$ with a classical signature scheme.

%% file: b_prelim.tex
\section{Preliminaries}

\subsection*{Subspaces}
\begin{itemize}
    \item For any subspace $A \leq \mathbb{Z}_2^n$, $A$ will also refer to a matrix whose columns are a basis of the subspace $A$. The matrix serves as a description of the subspace.
    \item Let $A^\perp = \{\bx \in \mathbb{Z}_2^n\ | \forall \mathbf{a} \in A, \langle \bx, \mathbf{a} \rangle = 0\}$ be the orthogonal complement of $A$.
    \item Let $\ket{A} = \frac{1}{\sqrt{|A|}} \sum_{\bx \in A} \ket{\bx}$
    \item Let $O_A: \mathbb{Z}_2^n \rightarrow \{0,1\}$ decide membership in $A$. That is, $\forall \bx \in \mathbb{Z}_2^n$: \[O_A(\bx) = \mathbbm{1}_{\bx \in A}\]
    Given a basis $B$ of $A^\perp$, we can compute $O_{A}$ as follows:
    \[O_A(\bx) = \mathbbm{1}_{B^T \cdot \bx = \mathbf{0}}\]
\end{itemize}

\subsection*{Quantum computation.} 
Here we recall the basics of quantum computation, and refer to Nielsen and Chuang~\cite{Nielsen2000} for a more detailed overview.

A quantum system is a Hilbert space $\Hs$ and an associated inner product $\langle\cdot|\cdot\rangle$.  The state of the system is given by a complex unit vector $|\psi\rangle$.  Given quantum systems $\Hs_1$ and $\Hs_2$, the joint quantum system is given by the tensor product $\Hs_1\otimes\Hs_2$.  Given $|\psi_1\rangle\in\Hs_1$ and $|\psi_2\rangle\in\Hs_2$, we denote the product state by $|\psi_1\rangle|\psi_2\rangle\in\Hs_1\otimes\Hs_2$.  A quantum state $|\psi\rangle$ can be ``measured'' in an  orthonormal basis $B=\{|b_0\rangle,...,|b_{d-1}\rangle\}$ for $\Hs$, which gives value $i$  with probability $|\langle b_i|\psi\rangle|^2$. The quantum state then collapses to the basis element $|b_i\rangle$.  

For a state over a joint system $\Hs_1\otimes\Hs_2$, we can also perform a partial measurement over just, say, $\Hs_1$. Let $\{|a_0\rangle,...\rangle\}$ be a basis for $\Hs_1$ and $\{|b_0\rangle,...\rangle\}$ a basis for $\Hs_2$. Then for a general state $|\psi\rangle=\sum_{i,j}\alpha_{i,j}|a_i\rangle|b_j\rangle$, measuring in $\Hs_1$ will give the outcome $i$ with probability $p_i=\sum_j |\alpha_{i,j}|^2$. In this case, the state collapses to $\sqrt{1/p_i}\sum_j \alpha_{i,j}|a_i\rangle|b_j\rangle$.

Operations on quantum states are given by unitary transformations over $\Hs$. An efficient quantum algorithm is a unitary $U$ that can be decomposed into a polynomial-sized circuit, consisting of unitary matrices from some finite set.

\subsection*{Miscellaneous}
A function $f(\lambda)$ is \textit{negligible}, written as $f(\lambda) = \negl(\lambda)$, if $f(\lambda) = o(\lambda^{-c})$ for any constant $c$. $\poly{}(\lambda)$ is a generic polynomial in $\lambda$. A probability $p$ is \textit{overwhelming} if $1 - p = \negl(\lambda)$. Finally $[\lambda] = \{1, \dots, \lambda\}$, for any $\lambda \in \Nbb$. Numbers are assumed to be in $\Nbb$ unless otherwise stated. \par

%% file: c_Definition.tex
\section{Definition of Franchised Quantum Money}
Here we'll define franchised quantum money and its notions of security in detail.
\begin{definition}[Main Variables]\label{def:vars}
\begin{itemize}
    \item Let $\lambda \in \Nbb$ be the security parameter.
    \item Let $N \in \Nbb$ be the number of verification keys that the bank distributes. $N = O(\poly(\lambda))$ in the security game because the adversary cannot query more than polynomially-many users.
    \item Let $C \in [N]$ be the collusion bound, the maximum number of verification keys that the adversary can receive.
    \item Let $msk$ be the master secret key, known only by the bank.
    \item Let $svk$ be a secret verification key given to a user.
    \item Let $\ket{\$}$ be a valid banknote. Let $\ket{P}$ be a purported banknote, which may or may not be valid.
    \item After verification, $\ket{\$}$ becomes $\ket{\$'}$, and $\ket{P}$ becomes $\ket{P'}$.
\end{itemize}
\end{definition}

\begin{definition} A \textbf{franchised quantum money scheme} $\mathcal{F}$ comprises four polynomial-time quantum algorithms: \setup{}, \franchise{}, \mint{}, and \ver{}.
\begin{enumerate}
    \item \textbf{\setup{}}: 
    The bank runs \setup{} to initialize the FQM scheme. 
    \[msk \leftarrow \setup{}(1^\lambda)\]

    \item \textbf{\franchise{}}: 
    The bank runs \franchise{} whenever a user requests a secret verification key. Then the bank sends $svk$ to the user.
    \[svk \leftarrow \franchise{}(msk)\]
    
    \item \textbf{\mint{}}: %$\ket{\$} \leftarrow \mint{}(msk)$ \\
    The bank runs \mint{} to create a new banknote $\ket{\$}$. Then the bank gives $\ket{\$}$ to someone who wants to spend it.
    \[\ket{\$} \leftarrow \mint{}(msk)\]
    
    \item \textbf{\ver{}}: 
    Any user with a secret verification key can run \ver{} to check whether a purported banknote $\ket{P}$ is valid. \ver{} accepts $\ket{P}$ ($b = 1$) or rejects $\ket{P}$ ($b=0$). Finally, $\ket{P}$ becomes $\ket{P'}$ after it is processed by \ver{}.
    \[b, \ket{P'} \leftarrow \ver{}(svk, \ket{P})\]
\end{enumerate}
\end{definition}

In order to function as money, $\ket{\$}$ should be accepted by \ver{} with overwhelming probability, and $\ket{\$'}$ should be close to $\ket{\$}$. This way, we can verify the state in future transactions. The following definition, for correctness, achieves these properties.

\begin{definition}\label{def:correctness}
$\mathcal{F}$ is \textbf{correct} if for any $svk \leftarrow \franchise{}(msk)$, any $\ket{\$} \leftarrow \mint(msk)$, and any $N$ and $C$ that are polynomial in $\lambda$,
\begin{enumerate}
    \item $\ver{}(svk, \ket{\$})$ accepts with probability overwhelming in $\lambda$, and
    \item The trace distance between $\ket{\$}$ and $\ket{\$'}$ is $\negl(\lambda)$.
\end{enumerate}

\end{definition}

Next, franchised quantum money needs two forms of security: security against counterfeiting and sabotage. Security against counterfeiting, defined below, means that an adversary given $m$ banknotes cannot produce $m+1$ banknotes that pass verification, except with $\negl(\lambda)$ probability.

\begin{definition}\label{def:security_count}
$\mathcal{F}$ is \textbf{secure against counterfeiting} if for any polynomial-time quantum adversary, the probability that the adversary wins the following security game is $\negl(\lambda)$:
\label{def:security}

\begin{enumerate}
    \item \textbf{Setup}: The challenger is given $\lambda, N, \text{ and } C$, where $N, C = \poly(\lambda)$. Then the challenger runs $\setup{}(1^\lambda)$ to get $msk$, and finally creates $N$ verification keys ($svk_1,
    \dots, svk_N$) by running $\franchise{}(msk)$ $N$ times.
    \item \textbf{Queries:} The adversary makes any number of franchise, mint, and verify queries, in any order:
    \begin{itemize}
        \item \textbf{Franchise}: the challenger sends a previously unused key to the adversary. By convention, let the last $C$ keys be sent to the adversary: $svk_{N-C+1}, \dots, svk_{N}$.
        \item \textbf{Mint}: The challenger samples $\ket{\$} \leftarrow \mint{}(msk)$ and sends $\ket{\$}$ to the adversary.
        \item \textbf{Verify}: The adversary sends a state $\ket{P}$ and an index $id \in [N-C]$ to the challenger. The challenger runs $\ver(svk_{id}, \ket{P})$, and sends the results $(b, \ket{P'})$ back to the adversary.
    \end{itemize}

    Let $m$ be the number of mint queries made, which represents the number of valid banknotes the adversary receives.\par

\item \textbf{Challenge}: The adversary tries to spend $m+1$ banknotes. The adversary sends to the challenger $u > m$ purported banknotes, possibly entangled, each with an $id \in [N-c]$: 
\[(id_1, \ket{P}_1), (id_2, \ket{P}_2), \dots, (id_u, \ket{P}_u)\]
Then for each purported banknote $\ket{P}_k$, the challenger runs \ver{}: 
\[b_k, \ket{P'}_k \leftarrow \ver{}(svk_{id_k}, \ket{P}_k)\]
The adversary wins the game if at least $m + 1$ of the purported banknotes are accepted.
\end{enumerate}
\end{definition}

The second form of security is security against sabotage. Sabotage is when the adversary tricks one user into accepting an invalid banknote that is then rejected by a second user.

\begin{definition}\label{def:security_sab}
$\mathcal{F}$ is \textbf{secure against sabotage} if for any polynomial-time quantum adversary, the probability that the adversary wins the following security game is $\negl(\lambda)$:
\begin{enumerate}
    \item \textbf{Setup:} same as in definition \ref{def:security_count}
    \item \textbf{Queries:} same as in definition \ref{def:security_count}
    \item \textbf{Challenge:} The adversary sends to the challenger a banknote $\ket{P}$ and two distinct indices $id_1, id_2 \in [N-c]$.
    
    The challenger runs $\ver{}$ using $svk_{id_1}$, then $svk_{id_2}$:
    \begin{align*}
        b_1, \ket{P'} &\leftarrow \ver{}(svk_{id_1}, \ket{P})\\
        b_2, \ket{P''} &\leftarrow \ver{}(svk_{id_2}, \ket{P'})
    \end{align*}
    The adversary wins the game if the first verification accepts ($b_1 = 1$) and the second verification rejects ($b_2 = 0$).
\end{enumerate}
\end{definition}

%% file: d_Simple_Construction.tex
\section{Simple Construction}
\label{sec:simple_construction}
Here we give a simpler version of our construction of FQM in order to illustrate the main ideas. The simple construction is correct and secure, but only if the adversary gets just one banknote. The full construction of FQM is given in section \ref{sec:construction}.

\subsection*{Variables and Parameters}
\begin{itemize}
    \item Let $N$ be any $\poly(\lambda)$.
    \item Let $n = \Omega(\lambda)$ be the dimension of the ambient vector space: $\Z_2^n$.
    \item Let $A < \Z_2^n$ be a subspace, and let $dim(A) = dim(A^\perp) = n/2$.
    \item Let $V \leq A$ and $W \leq A^\perp$ be two subspaces given by an svk.
    \item Let $t = \Theta(\sqrt{n})$ be an upper bound on the dimension of $V$ and $W$.
    \item Let $C = \frac{n}{4t}$.
    
\end{itemize}

\subsection*{$\setup{}$}
\textbf{Input:} $1^\lambda$
\begin{enumerate}
    \item Choose values for $N, n, \text{ and } t$.
    \item Sample $A \leq \mathbb{Z}_2^n$ such that $dim(A) = dim(A^\perp) = n/2$.
    \item For each $id \in [N]$: sample $t$ indices uniformly and independently from $[n/2]$. Call this set $I_{id}$. Then sample another set called $J_{id}$ from the same distribution.
    \item Sample $\bv_1, \dots, \bv_{n/2} \in A$ independently and uniformly at random.\\
    Sample $\bw_1, \dots, \bw_{n/2} \in A^\perp$ independently and uniformly at random.
    \item
    \[\text{Let } msk = \Big(A, \{\bv_i\}_{i \in [n/2]}, \{\bw_j\}_{j \in [n/2]}, \{I_{id}, J_{id}\}_{id \in [N]}\Big)\]
    and \textbf{output} $msk$.
\end{enumerate}

\subsection*{\franchise{}}
\textbf{Input:} msk
\begin{enumerate}
    \item Choose an $id \in [N]$ that hasn't been chosen before.
    \item Let $svk_{id} = \big(I_{id}, J_{id}, \{\bv_i\}_{i \in I_{id}}, \{\bw_j\}_{j \in J_{id}}\big)$, and \textbf{output} $svk_{id}$.
\end{enumerate}

\subsection*{\mint{}}
\textbf{Input:} msk
\begin{enumerate}
    \item Generate and \textbf{output} $\ket{\$} = \ket{A}$.
\end{enumerate}

\subsection*{\ver{}}
\textbf{Input:} $svk, \ket{P}$\\

Let $svk = \big(I, J, \{\bv_i\}_{i \in I}, \{\bw_j\}_{j \in J}\big)$. Then let 
\[V := span(\{\bv_i\}_{i \in I}) \text{ and } W = span(\{\bw_j\}_{j \in J})\]
\begin{enumerate}
    \item \textbf{Computational basis test:} Check that $O_{W^\perp}\big(\ket{P}\big) = 1$. Now $\ket{P}$ becomes $\ket{P_1}$.
    \item Take the quantum Fourier transform of $\ket{P_1}$ to get $\widetilde{\ket{P_1}}$.
    \item \textbf{Fourier basis test:} Check that $O_{V^\perp}\big(\widetilde{\ket{P_1}}\big) = 1$. Now $\widetilde{\ket{P_1}}$ becomes $\widetilde{\ket{P_2}}$.
    \item Take the inverse quantum Fourier transform of $\widetilde{\ket{P_2}}$ to get $\ket{P_2}$. Let $\ket{P'} = \ket{P_2}$. \textbf{Output} $1$ (accept) if both tests pass, and $0$ (reject) otherwise. Also output $\ket{P'}$.
\end{enumerate}

\subsection*{Proofs of Correctness and Security}

\begin{theorem}
The simple FQM construction is correct.
\end{theorem}
\begin{proof}
We will show that for any valid banknote $\ket{\$} = \ket{A}$, $\ver{}(svk, \ket{\$})$ outputs $(1,\ket{\$})$ with probability $1$.
\begin{enumerate}
    \item The computational basis test passes with probability $1$. $W \leq A^\perp$, so $A \leq W^\perp$, and $O_{W^\perp}(\ket{A}) = 1$ with probability $1$. Also the banknote is unchanged by this test.
    \item The quantum Fourier transform of the banknote is $\ket{A^\perp}$ (\cite{AC12}).
    \item The Fourier basis test also passes with probability $1$.
    Since $V \leq A$, then $A^\perp \leq V^\perp$, and $O_{V^\perp}(\ket{A^\perp}) = 1$ with probability $1$. The banknote is also unchanged by this test.
    \item Finally, the inverse quantum Fourier transform restores the banknote to its initial state $\ket{A}$, and the banknote is accepted by $\ver{}$ with probability $1$.
\end{enumerate}
\end{proof}

\begin{theorem}
\label{thm:simple_security}
The simple FQM construction is secure against counterfeiting if the adversary receives only $m=1$ banknote.
\end{theorem}

\begin{proof}

\quad\\ 
\noindent 1) Preliminaries\\
Let's say without loss of generality that the adversary receives $C$ verification keys, which correspond to the last $C$ identities: $id \in \{N-C+1, \dots, N\}$. Then they receive $1$ banknote, and then they make any polynomial number of verification queries. Finally, they attempt the counterfeiting challenge.

We can define the subspaces $V_{adv} \leq A$ and $W_{adv} \leq A^\perp$ as the subspaces known to the adversary. We also define $V_{id}$ and $W_{id}$ analogously for each $id \in [N]$:

\begin{definition}
\quad
\begin{itemize}
    \item Let $I_{adv} = \bigcup_{id > N-C} I_{id} \text{ and } J_{adv} = \bigcup_{id > N-C} J_{id}$.
    \item For any $id \in [N]$, let $V_{id} = span\big(\{\bv_i\}_{i \in I_{id}}\big)$. Let $W_{id}$, $V_{adv}$, and $W_{adv}$ be defined analogously.
\end{itemize}
\end{definition}

Let's assume for simplicity that
\[dim(V_{adv}) = dim(W_{adv}) =: d\]
where $d$ is fixed.  This assumption isn't necessary for proving security, but it does make the proof simpler. Also note that $d \leq n/4$.

\quad\par
\noindent 2) We'll use a hybrid argument to reduce the counterfeiting game to \cite{AC12}'s security game for secret key quantum money:
\begin{itemize}
    \item \textbf{h0} is the counterfeiting security game for the simple FQM construction. In particular, the adversary receives one banknote $\ket{A}$, along with $C$ franchised verification keys.
    \item \textbf{h1} is the same as h0, except the challenger simulates full verifiers: whenever the adversary makes a verification query $(id, \ket{P})$, the challenger verifies the state using $O_A$ and $O_{A^\perp}$ instead of $O_{W_{id}^\perp}$ and $O_{V_{id}^\perp}$.
    \item \textbf{h2} is essentially \cite{AC12}'s security game for secret key quantum money: let $A \leq \mathbb{Z}_2^{n-2d}$ be a uniformly random  subspace such that $dim(A) = dim(A^\perp) = n/2 - d$. Next, the adversary gets a banknote $\ket{A}$ but no verification keys. They can make verification queries, and the challenger will run $\ver{}$ using full verifiers: ($O_A$ and $O_{A^\perp}$).
\end{itemize}
    \begin{lemma}
    \label{thm:hybrid_indist}
    For any polynomial-time adversary $\mathcal{A}$, their success probabilities in $h0$ and in $h1$ differ by a $\negl(\lambda)$ function. 
    \end{lemma}
    We'll defer the proof of lemma \ref{thm:hybrid_indist} to section \ref{sec:dist_game}.
    
    \begin{lemma}\label{thm:h2_h3}
    If $\mathcal{A}$ is a polynomial-time adversary with non-negligible success probability in $h1$, then there is a polynomial-time adversary $\mathcal{A}'$ with non-negligible success probability in $h2$.
    \end{lemma}
    \begin{proof}
    We can reduce the security game in $h2$ to the security game in $h1$. Let $\mathcal{A}'$ be given an $h2$ banknote $\ket{A}$, where $A \leq \mathbb{Z}_2^{n-2d}$ and $dim(A) = dim(A^\perp) = n/2 - d$. We will turn $\ket{A}$ into an $h1$ banknote $\ket{B}$, where $B \leq \mathbb{Z}_2^n$, and $dim(B) = dim(B^\perp) = n/2$:
    \begin{enumerate}
        \item Prepend $\ket{A}$ with $\ket{0}^{\otimes d}\ket{+}^{\otimes d}$:
        \[\text{Let } \ket{A'} = \ket{0}^{\otimes d}\ket{+}^{\otimes d}\ket{A}\]
        $\ket{A'}$ is a subspace state, a uniform superposition over the subspace 
        \[A' := span[\hat{e}_{d+1}, \dots, \hat{e}_{2d}, (0^{\times 2d} \times A)]\]
        where $0^{\times 2d} \times A$ is all vectors in $\mathbb{Z}_2^n$ for which the first $2d$ bits are $0$ and the rest form a vector in $A$.
        Also, $dim(A') = dim(A'^\perp) = n/2$. 
        \item Sample an invertible matrix $M \in \mathbb{Z}_2^{n \times n}$ uniformly at random. Then apply $M$ to $\ket{A'}$:
        \[\text{Let } B = M \cdot A' \text{ and } \ket{B} = M(\ket{A'})\]
    Observe that $\ket{B}$ is a uniformly random $h1$ banknote.
    \end{enumerate}    
    
    Additionally, the adversary knows $d$ dimensions of $B$ and $d$ dimensions of $B^\perp$:
    \begin{align*}
        V_{adv} &= M \cdot span(\hat{e}_{d+1}, \dots, \hat{e}_{2d})\\
        W_{adv} &= M \cdot span(\hat{e}_{1}, \dots, \hat{e}_{d})
    \end{align*}
    $\mathcal{A}'$ derives $C$ $h1$-verification keys whose vectors span $V_{adv}$ and $W_{adv}$. 
    Finally, $\mathcal{A}'$ runs $\mathcal{A}$, giving it the banknote $\ket{B}$ along with the verification keys.
    
    When $\mathcal{A}$ makes a verification query $(id, \ket{P})$, $\mathcal{A}'$ simulates the $h1$ challenger's response as follows, by converting $\ket{P}$ into an $h2$ banknote:
    \begin{enumerate}
        \item Let $\ket{P'} = M^{-1}(\ket{P})$.
        \item Check that the first $2d$ qubits of $\ket{P'}$ are $\ket{0}^{\otimes d} \ket{+}^{\otimes d}$.
        \item Query the $h2$ challenger with the remaining $n-2d$ qubits of $\ket{P'}$. Let $\ket{P''}$ be the state returned by the challenger. Accept the banknote if and only if the first $2d$ qubits passed their test, and the challenger accepted as well.
        \item Return $M(\ket{0}^{\otimes d} \ket{+}^{\otimes d} \ket{P''})$ to the $h1$ adversary.
    \end{enumerate}
    
    This procedure simulates $h1$ for $\mathcal{A}$. Also, note that the probability that $\ket{P'}$ passes $h2$ verification is at least the probability that $\ket{P}$ passes $h1$ verification.
    
    Finally, when $\mathcal{A}$ attempts to win the challenge by outputting several purported $h1$ banknotes, $\mathcal{A}'$ converts these into $h2$ banknotes. If $\mathcal{A}$ wins in $h1$ with non-negligible probability, then $\mathcal{A}'$ wins in $h2$ with at least that probability.
    \end{proof}

    \begin{lemma}
    \label{thm:h3}
    In $h2$, any polynomial-time adversary has negligible success probability.
    \end{lemma}
    \begin{proof}
    \cite{AC12}'s security game is similar to $h2$, except the adversary can query both $O_{A}$ and $O_{A^\perp}$. They proved the following:
    \begin{theorem}[\cite{AC12}, Theorem 25]
    Let the adversary get $\ket{A}$, a random $n'$-qubit banknote, along with quantum query access to $O_{A}$ and $O_{A^\perp}$. If the adversary prepares two possibly entangled banknotes that both pass verification with probability $\geq \varepsilon$, for all $1/\varepsilon = o(2^{n'/2})$, then they make at least $\Omega(\sqrt{\varepsilon} 2^{n'/4})$ oracle queries. 
    \end{theorem}
    Let $n' = n-2d$, the size of the banknote in $h2$. Note that $n' \geq n/2$. Next, let $\varepsilon = 2^{-n'/3}$. Note that $\varepsilon = \negl(\lambda)$. Finally, the number of queries needed to win with probability $\geq \epsilon$ is
    \[\Omega(\sqrt{\varepsilon} 2^{n'/4}) = \Omega(2^{n'/4 - n'/6}) = \Omega(2^{n'/12})\]
    
    Any polynomial-time adversary makes fewer than that many queries, so no polynomial-time adversary can win with non-negligible probability.
    \end{proof}
    Putting together lemmas \ref{thm:hybrid_indist}, \ref{thm:h2_h3}, \ref{thm:h3}, we get that any polynomial-time adversary has negligible probability of winning the counterfeiting security game for the simple construction of FQM.
\end{proof}

\begin{theorem}
\label{thm:simple_security_sabotage}
The simple FQM construction is secure against sabotage if the adversary receives only $m=1$ banknote.
\end{theorem}
\begin{proof}
The proof of this theorem follows the proof of \ref{thm:simple_security}, except at the end. We need to show that in $h2$, any polynomial-time adversary has negligible probability of succeeding at sabotage. To show this, we need the following lemma:

\begin{lemma}[\cite{AC12}, Lemma 21]
In $h2$, $\ver{}$ projects $\ket{P}$ onto $\ket{A}$ if it accepts and onto a state orthogonal to $\ket{A}$ if it rejects.
\end{lemma}
That means that if a purported banknote is verified twice, it is either accepted both times or rejected both times. Therefore, sabotage is not possible in $h2$.

Again, by lemmas \ref{thm:hybrid_indist} and \ref{thm:h2_h3}, any polynomial-time adversary has negligible probability of winning the sabotage security game for the simple construction of FQM.
\end{proof}

%% file: f_Full_Construction.tex
\section{Full Construction}\label{sec:construction}
The full construction of FQM adds a signature scheme and a secret key encryption scheme, which let us hand out the subspaces $V_{id}, W_{id}$ as part of the banknote. As a result, a user can verify many banknotes, each for a different subspace $A$, without needing to call $\franchise{}$ for each banknote. \\

\noindent The signature and encryption schemes have the following syntax.
\begin{definition}(\cite{KL14}, Definition 12.1)
A \textbf{signature scheme} comprises the following three probabilistic polynomial-time algorithms: 
\begin{itemize} 
    \item $\sigkeygen{}$ takes a security parameter $\lambda$, and returns $(sig\_pk, sig\_sk)$, the public and secret keys.
    \[sig\_pk, sig\_sk \leftarrow \sigkeygen{}(1^\lambda)\]
    \item $\sign{}$ takes a message $msg \in \{0,1\}^*$ and the secret key and produces $\sigma$, the signature for $msg$.
    \[\sigma \leftarrow \sign{}(sig\_sk, msg)\]
    \item $\sigver{}$ takes $msg$, $\sigma$, and the public key, and outputs a bit $b$ to indicate the decision to accept ($b=1$) or reject ($b=0$) the signature-message pair. Also, \sigver{} is deterministic.
    \[b := \sigver{}(sig\_pk, msg, \sigma)\]
\end{itemize}
\end{definition}
The signature scheme is \textit{existentially unforgeable under an adaptive chosen-message attack}. Such a signature scheme can be constructed from one-way functions (\cite{KL14}).

\begin{definition}(\cite{KL14}, Definition 3.7).
A \textbf{secret key encryption scheme} comprises the following three probabilistic polynomial-time algorithms:
\begin{itemize}
    \item $\enckeygen$ takes a security parameter $\lambda$ and produces a secret key $enc\_k$.
    \[enc\_k \leftarrow \enckeygen(1^\lambda)\]
    \item $\enc$ encrypts a message $msg \in \{0,1\}^*$ using the key $enc\_k$ to produce a cyphertext $c$.
    \[c \leftarrow \enc(enc\_k, msg)\]
    \item $\dec$ decrypts $c$, again using $enc\_k$. \dec{} is deterministic, so for any $enc\_k$ produced by \enckeygen{}, \dec{} always decrpyts $c$ correctly.
    \[msg := \dec(enc\_k, c)\]
\end{itemize}
\end{definition}
The secret key encryption is \textit{CPA-secure}, and it can also be constructed from one-way functions (\cite{KL14}).

\subsubsection*{Variables}
\begin{itemize}
    \item Let $\ket{\$}$, a valid banknote, comprise a quantum state $\ket{\Sigma}$ and some classical bits.
    \item Let $\ket{P}$, a purported banknote, comprise a quantum state $\ket{\Pi}$ and some classical bits.
%
\end{itemize}

\subsubsection*{\setup{}}\quad\\

\noindent\textbf{Input:} $1^\lambda$
\begin{enumerate}
    \item Choose values for the parameters: $n =\Omega(\lambda), t = \Theta(\sqrt{n})$.
    \item Set up one signature scheme and $n$ encryption schemes by computing: 
    \begin{align*}
        (sig\_pk, sig\_sk) &\leftarrow \sigkeygen{}(1^\lambda)\\
        (enc\_k_1, \dots, enc\_k_{n}) &\leftarrow \enckeygen(1^\lambda), \dots, \enckeygen(1^\lambda)
    \end{align*}
    \item Let $msk = (sig\_pk, sig\_sk, enc\_k_1, \dots, enc\_k_{n})$, and then \textbf{output} $msk$.
\end{enumerate}

\subsubsection*{\franchise{}}\quad\\

\noindent\textbf{Input:} $msk$
\begin{enumerate}
    \item Sample $t$ indices uniformly and independently from $[n/2]$. Call this set $I$. Then sample another set called $J$ from the same distribution.
    \item Let $svk = (sig\_pk, I, J, \{enc\_k_{i}\}_{i\in I}, \{enc\_k_{j + n/2}\}_{j\in J})$, and then \textbf{output} $svk$.
\end{enumerate}

\subsubsection*{\mint{}}\quad\\

\noindent\textbf{Input:} $msk$
\begin{enumerate}
    \item Sample a subspace $A < \Z_2^n$ such that $dim(A) = dim(A^\perp) = n/2$, uniformly at random. 
    \item Create the subspace state for $A$, and let $\ket{\Sigma} = \ket{A}$.
    \item Sample $n/2$ random vectors in $A$: $\{\bv_1, \dots, \bv_{n/2}\} \in_R A$. And sample $n/2$ random vectors in $A^\perp$: $\{\bw_1, \dots, \bw_{n/2}\} \in_R A^\perp$.
    \item Encrypt the $\bv$s and $\bw$s, each with a different $enc_k$:
    \begin{align*}
        \text{Let } c_1, \dots, c_{\frac{n}{2}} &= \big[\enc(enc\_k_{1}, \bv_1), \dots, \enc(enc\_k_{\frac{n}{2}}, \bv_{\frac{n}{2}})\\
        c_{\frac{n}{2}+1}, \dots, c_{n} &= \big[\enc(enc\_k_{\frac{n}{2}+1}, \bw_1), \dots, \enc(enc\_k_{n}, \bw_{\frac{n}{2}})
    \end{align*}
    \item Sign the ciphertexts. Let $\sigma \leftarrow \sign[sig\_sk, (c_1, \dots, c_{n})]$.
    \item Construct the banknote. Let $\ket{\$} = (\ket{\Sigma}, c_1, \dots, c_{n}, \sigma)$.
    Finally, \textbf{output} $\ket{\$}$.
\end{enumerate}

\subsubsection*{\ver{}}\quad\\

\noindent\textbf{Inputs:} $svk_{id}, \ket{P}$
\begin{enumerate}
    \item Check the signature: $\sigver{}(sig\_pk, (c_1, \dots, c_{n}), \sigma)$.
    \item Decrypt any ciphertexts for which the key is available. For every $i \in I_{id}$ compute $\bv_i = \dec(enc\_k_i, c_i)$, and for every $j \in J_{id}$, compute $\bw_j = \dec(enc\_k_{j+ n/2}, c_{j+ n/2})$.\\
    
    Additionally, define two subspaces, $V_{id}, W_{id}$:
    \begin{align*}
        V_{id} &:= span(\{\bv_{i}\}_{i\in I_{id}})\\
        W_{id} &:= span(\{\bw_{j}\}_{j \in J_{id}})
    \end{align*}
        
    \item Recall that $\ket{P}$ comprises a quantum state $\ket{\Pi}$ and some classical bits.\\
    \textbf{Computational basis test:} Check that $O_{W_{id}^{\perp}}(\ket{\Pi}) = 1$. After this step, $\ket{\Pi}$ becomes $\ket{\Pi_1}$.
    \item Take the quantum Fourier transform of $\ket{\Pi_1}$ to get $\widetilde{\ket{\Pi_1}}$.
    \item \textbf{Fourier basis test:} Check that $O_{V_{id}^\perp}(\widetilde{\ket{\Pi_1}}) = 1$. After this step, $\widetilde{\ket{\Pi_1}}$ becomes $\widetilde{\ket{\Pi_2}}$.
    \item Take the inverse quantum Fourier transform of $\widetilde{\ket{\Pi_2}}$ to get $\ket{\Pi_2}$. \par 
    
    Let $\ket{P'}$ be the state that $\ket{P}$ has become, with $\ket{\Pi}$ replaced with $\ket{\Pi_2}$. \par
    
    \textbf{Output} $1$ (accept) if both tests pass, and $0$ (reject) otherwise. Also output $\ket{P'}$.
\end{enumerate}

\subsection*{Proofs of Correctness and Security}
\begin{theorem}\label{thm:correctness}
The full construction of franchised quantum money is \textit{correct}.
\end{theorem}
\begin{proof}
In steps 1 and 2 of \ver, we check the signature and decrypt the ciphertexts. With probability $1$, the signature check passes, and the ciphertexts are correctly decrypted. This follows from the correctness of the signature and encryption schemes.

After the first two steps, \ver{} is the same as it was in the simple construction. Because the simple construction is correct, the full construction is correct as well.
\end{proof}

\begin{theorem}\label{thm:counterfeiting}
The full construction of franchised quantum money is secure against counterfeiting and sabotage.
\end{theorem}
\begin{proof} We will use a hybrid argument to show that the adversary's success probability at counterfeiting or sabotage with the full construction is close to what it is with the simple construction. Since the simple construction is secure against counterfeiting and sabotage, the full construction is secure as well.\\

\noindent1) Preliminaries\\
Without loss of generality, let us say that the adversary receives $C$ $svk$s, then receives $m$ valid banknotes from the challenger, and finally makes multiple \ver{} queries. \par

Furthermore, let the challenger keep a record of all the banknotes and $svk$s it generated. 
Finally let the ciphertexts $(c_1, \dots, c_n)$ of each valid banknote be unique. This occurs with overwhelming probability.\\

\noindent 2) Next, we'll use a sequence of hybrids to simplify the situation and remove the need for the signature and encryption schemes. 
\begin{itemize}
    \item $\mathbf{h0}$ uses the full FQM construction in the counterfeiting or sabotage security game.
    \item $\mathbf{h1}$ is the same as $h0$, except \ver{} only accepts a purported banknote if its ciphertexts $(c_1, \dots, c_n)$ match those of one of the $m$ valid banknotes given to the adversary.
    \item $\mathbf{h2}$ is the same as $h1$, except for any ciphertext $c_i$ for which the adversary does not have the decryption key, $c_i$ is replaced with junk: the encryption under $enc\_k_i$ of a random message.
\end{itemize}
        
The adversary has $\negl(\lambda)$ advantage in distinguishing $h0$ and $h1$. The signature scheme is existentially unforgeable under an adaptive chosen-message attack, so except with $\negl(\lambda)$ probability, any banknote that passed \ver{} in $h0$ had ciphertexts that matched one of the $m$ valid banknotes.\par
        
The adversary has $\negl(\lambda)$ advantage in distinguishing $h1$ and $h2$ because the encryption scheme is CPA-secure. For any $i$ for which the adversary does not have the decryption key, the adversary receives either $m$ ciphertexts of random messages or $m$ ciphertexts of potentially useful messages. CPA security is equivalent to left-or-right security (\cite{KL14}), which implies that the adversary cannot distinguish these two cases.\\

\noindent 3) Next, we'll use another set of hybrids to relate the full construction with the simple construction.
\begin{itemize}
    \item $\mathbf{h3}$ is the same as $h2$, except we do not use the signature or encryption schemes. Each valid banknote comprises a subspace state $\ket{\psi_A}$ and a set of plaintext $\bv$ vectors in $A$ and $\bw$ vectors in $A^\perp$. 
    Finally, to verify a purported banknote, the challenger checks that the $\bv$ and $\bw$ vectors associated with a purported banknote match those of a valid banknote. Then they use whatever $svk$s were recorded along with the valid banknote to verify the subspace state.
    
    \item $\mathbf{h4}$ is the simple FQM construction with just one banknote. This is the same as $h3$, except the adversary receives only $1$ valid banknote, and the $\bv$ and $\bw$ vectors are given by $\franchise{}$ and are not included with the banknote.
\end{itemize}

The adversary's best success probability is the same in h2 and h3 because the signature and encryption schemes were not necessary in h2, so h3 presents essentially the same security game to the adversary.

\begin{lemma}
The best success probability for an adversary in h3 is at most $m$ times the best success probability in h4.
\end{lemma}
\begin{proof}
\label{thm:h3_h4}
Given any $h3$ adversary $\mathcal{A}$, there is an $h4$ adversary $\mathcal{A}'$ that simulates $\mathcal{A}$. $\mathcal{A}'$ receives one valid banknote and generates $m-1$ other banknotes. Then $\mathcal{A}'$ runs $\mathcal{A}$ with the $m$ banknotes. When $\mathcal{A}$ makes a verification query, $\mathcal{A}'$ simulates the verifier for the $m-1$ banknotes it generated and queries the $h4$ verifier for the banknote that it received. Finally, $\mathcal{A}$ outputs some purported banknotes at the challenge step, which $\mathcal{A}'$ outputs as well. \par
        
If $\mathcal{A}$ wins in $h3$, then there are at least $m+1$ purported banknotes that pass verification, and at least two of them have the same $\bv$ and $\bw$ vectors. $\mathcal{A}'$ wins in $h4$ if the two banknotes with matching vectors also match the vectors of the banknote given to $\mathcal{A}'$. This happens with probability $\frac{1}{m}$, by the symmetry of the $m$ banknotes. Therefore, $\mathcal{A}'$'s success probability is $\frac{1}{m}$ times $\mathcal{A}$'s.
\end{proof}

\noindent 4) In $h4$, the adversary has negligible probability of winning the counterfeiting or sabotage games, by theorems \ref{thm:simple_security} and \ref{thm:simple_security_sabotage}. Since $m = O(\poly(\lambda))$, for any polynomial-time adversary, then any polynomial-time adversary has negligible probability of winning the counterfeiting or security games for the full FQM construction.
\end{proof}

%% file: e_Dist_Game.tex
\section{Distinguishing Game}
\label{sec:dist_game}
In order to prove lemma \ref{thm:hybrid_indist}, we will use the adversary method of \cite{Amb00}. We will study the \textit{distinguishing game}, in which an adversary that is more powerful than the one in lemma \ref{thm:hybrid_indist} tries to distinguish full and franchised verifiers. Then we show that the more-powerful adversary still has negligible advantage.

In the distinguishing game, the adversary is given a classical description of $A$, along with other information that is more than what they receive in the security game. However, one piece of information remains hidden to them: the verification keys used by the franchised verifiers. More formally, we say the adversary is given the $msk$, which includes every $(V_{id}, W_{id})$. But the verifiers will actually use $(M \cdot V_{id}, M \cdot W_{id})$ for some random matrix $M$. The next two definitions make this precise.
\begin{definition}
Let $\mathcal{M}(A)$ be the set of all matrices $M \in \mathbb{Z}_2^{n \times n}$ such that:
\begin{itemize}
    \item $M$ is invertible
    \item If $\bx \in A$, then $M^T\bx \in A$, and if $\bx \in A^\perp$, then $M^T\bx \in A^\perp$.
\end{itemize}
\end{definition}

\begin{definition}
For any $M \in \mathcal{M}(A)$, we also treat $M$ as a function mapping one master secret key to another. Essentially, $M$ is applied to every $\bv$ or $\bw$ vector that the adversary did not receive. More formally, for any $msk$:
\[M(msk) = \Big(A, \{\bv_i\}_{i \in I_{adv}}, \{M \cdot \bv_i\}_{i \not\in I_{adv}}, \{\bw_j\}_{j \in J_{adv}}, \{M \cdot \bw_j\}_{j \not\in J_{adv}}, \{I_{id},J_{id}\}_{id\in[N]}\Big)\]
\end{definition}
\par

Let $msk' = M(msk)$, and let $V_{adv}', W_{adv}', V_{id}',$ and $W_{id}'$ be defined analogously. Then $V_{adv}' = V_{adv}$ and $W_{adv}' = W_{adv}$ because the adversary's $\bv$ and $\bw$ vectors are not changed by $M$. Therefore, in the counterfeiting and sabotage security games, the adversary receives the same information, whether the master secret key is $msk$ or $msk'$.\\

\noindent Next, the adversary in the distinguishing game can also query $O_{W_{id}^\perp}$ and $O_{V_{id}^\perp}$, rather than just $\ver{}$. The following definitions bundle together the oracles that the adversary can query.
\begin{definition}
The \textbf{franchised verification oracle} for a given $msk$ is $O_{Fran}[msk]$. It takes as input an $id \in [N-C]$, a selection bit $s \in \{0,1\}$, and a vector $\bx \in \mathbb{Z}_2^n$. Then
\[O_{Fran}[msk](id, s, \bx) = \begin{cases}
O_{W_{id}^\perp}(\bx) & s = 0\\
O_{V_{id}^\perp}(\bx) & s = 1
\end{cases}\]
\end{definition}

\begin{definition}
The \textbf{full verification oracle} for a given $msk$ is $O_{Full}[msk]$ or $O_{Full}[A]$. It takes as input $id \in [N-C]$, $s \in \{0,1\}$, and $\bx \in \mathbb{Z}_2^n$. Then
\[O_{Full}[A](id, s, \bx) = \begin{cases}
O_{A}(\bx) & s = 0\\
O_{A^\perp}(\bx) & s = 1
\end{cases}\]
\end{definition}

\noindent Now we can define the distinguishing game precisely.
\begin{definition}
The \textbf{distinguishing game} takes as input an $msk$, which is given to the challenger and the adversary. Then:
\begin{enumerate}
    \item The challenger samples $b \in_R \{0,1\}$ and $M \in_R \mathcal{M}(A)$.
    \item The adversary makes quantum queries to the challenger. If $b = 0$, the challenger uses $O_{Full}[A]$ to answer the queries; if $b=1$, the challenger uses $O_{Fran}[M(msk)]$.
    \item The adversary outputs a bit $b'$, and they win if and only if $b' = b$.
\end{enumerate}
\end{definition}

\begin{theorem}
\label{thm:dist_adv}
Any polynomial-time quantum adversary $\mathcal{A}$ has negligible advantage in the distinguishing game. That is: 
\[\Big|P[\mathcal{A} = 1 | b=0] - P[\mathcal{A} = 1 | b=1] \Big| \leq \negl(\lambda)\]
where the probabilities are over the choice of $M \in \mathcal{M}(A)$ and $\mathcal{A}$'s randomness.
\end{theorem}
We'll prove theorem \ref{thm:dist_adv} later using the adversary method, but assuming theorem \ref{thm:dist_adv} for now, we can prove lemma \ref{thm:hybrid_indist}.

\subsection*{Proof of lemma \ref{thm:hybrid_indist}}
We want to show that for any polynomial-time adversary $\mathcal{A}$, their success probabilities in $h0$ and in $h1$ differ by a $\negl(\lambda)$ function. Recall that $h0$ uses franchised verifiers, whereas $h1$ uses full verifiers.

Assume toward contradiction that $\mathcal{A}$'s success probabilities in $h0$ and $h1$ differ by a non-negligible amount. Then we can construct an adversary $\mathcal{A}'$ that has non-negligible advantage in the distinguishing game.
    
$\mathcal{A}'$ simulates the counterfeiting security game and runs $\mathcal{A}$ on it. Given $msk$, $\mathcal{A}'$ constructs $\ket{A}$ and the $C$ franchised verification keys. When $\mathcal{A}$ queries a verifier, $\mathcal{A}'$ simulates this by querying either $O_{Full}[A]$ (if we're in $h1$) or $O_{Fran}[M(msk)]$ (if we're in $h0$). $\mathcal{A}'$ can even simulate the counterfeiting challenge, checking if $\mathcal{A}$ successfully counterfeited. Finally, $\mathcal{A}'$ outputs $1$ if $\mathcal{A}$ won the security game, and $0$ otherwise. $h0$ and $h1$ for the counterfeiting game correspond to $b=1$ and $b=0$ in the distinguishing game, so $\mathcal{A}'$ has non-negligible advantage in the distinguishing game.
    
This is a contradiction, by theorem \ref{thm:dist_adv}, so in fact, the success probabilities of $\mathcal{A}$ in the two hybrids must be negligibly close.

\subsection*{The Adversary Method}
Now we'll prove theorem \ref{thm:dist_adv} using the adversary method\footnote{Our proof is inspired by \cite{AC12}.}. First, we'll define the scenario that \cite{Amb00} considered, which is an abstract version of the distinguishing game, and then we'll state their main theorem. 
\begin{definition}
\label{def:adversary_method}
Let $\mathcal{O}$ be a set of oracles, each of which has range $\{0,1\}$. Let $f: \mathcal{O} \rightarrow \{0,1\}$ be a predicate that takes an oracle as input. Let $X, Y$ partition $\mathcal{O}$ such that $f(O_x) = 0$, for all $O_x \in X$, and $f(O_y) = 1$, for all $O_y \in Y$. \par
\end{definition}

Next, the adversary will try to compute $f$ on every input, so it must distinguish oracles in $X$ from oracles in $Y$.

\begin{definition}
\label{def:approx_compute}
Let $\mathcal{A}^O$ be a quantum algorithm with query access to an $O \in \mathcal{O}$. We say that $\mathcal{A}$ \textbf{approximately computes} $f$ if for every $O \in \mathcal{O}$, $P[\mathcal{A}^O = f(O)] \geq 2/3$.
\end{definition}

\begin{definition}
Let $u, u'$ be upper bounds that satisfy: 
\begin{itemize}
    \item For any $O_x \in X$ and any input $i$ to $O_x$, $P_{O_y \in Y}[O_x(i) \neq O_y(i)] \leq u$.
    \item For any $O_y \in Y$ and any input $i$ to $O_y$, $P_{O_x \in X}[O_x(i) \neq O_y(i)] \leq u'$.
\end{itemize}
\end{definition}

\begin{theorem}[\cite{Amb00}, Thm. 2]
\label{thm:adv_method}
If $\mathcal{A}$ approximately computes $f$, then $\mathcal{A}$ makes at least $\Omega \Big(\frac{1}{\sqrt{u \cdot u'}}\Big)$ queries to $O$.
\end{theorem}

\subsection*{Proof of theorem \ref{thm:dist_adv}}
The distinguishing game's format matches the format considered by the adversary method. For a given $msk$, let $X$ comprise only the full verification oracle, $\{O_{Full}[A]\}$. Let $Y$ comprise all possible franchised verification oracles: $Y = \{O_{Fran}[M(msk)] | M \in \mathcal{M}(A)\}$. And let $\mathcal{O} = X \bigcup Y$. Then $f$ equals $b$ from the distinguishing game.\\

\noindent Next, we will assume that each honest verifier gets at least $t/4$ dimensions of $V_{id}$ and $t/4$ dimensions of $W_{id}$ that are unknown to the adversary. As a result, each verifier accepts a negligible fraction of the vectors in $\mathbb{Z}_2^n$. So it is hard for the adversary to find an $\bx \in \mathbb{Z}_2^n$ on which the full and franchised oracles behave differently, which makes distinguishing them hard. The next definition and next two lemmas expand on this argument.

\begin{definition}
An $msk \leftarrow \setup(1^\lambda)$ is \textbf{good} if for every $id \in [N-C]$, 
\begin{itemize}
    \item $dim[span(V_{adv}, V_{id})] \geq dim(V_{adv}) + t/4$
    \item $dim[span(W_{adv}, W_{id})] \geq dim(W_{adv}) + t/4$
\end{itemize}
\end{definition}

\begin{lemma}
\label{thm:good}
With overwhelming probability in $\lambda$, $msk \leftarrow \setup(1^\lambda)$ is good.
\end{lemma}
\begin{proof}
\quad\\

\noindent 1) With overwhelming probability, $|I_{id}\backslash I_{adv}| \geq t/4$ for all $id \in [N-C]$.\\
First, $|I_{adv}| \leq C t = n/4$, so the probability that a uniformly random $i \in [n/2]$ is in $I_{adv}$ is $\leq 1/2$. Then
\[\text{Let } \mu = \mathbb{E}_{I_{id}}[|I_{id}\backslash I_{adv}|] \geq t/2\]
Next we use the multiplicative Chernoff bound:
\begin{align*}
    P\big[|I_{id}\backslash I_{adv}| \leq t/4\big] &\leq P\big[|I_{id}\backslash I_{adv}| \leq \mu/2\big]\\
    &< \bigg(\frac{e^{-1/2}}{(1/2)^{1/2}}\bigg)^{\mu} = \Big(\frac{2}{e}\Big)^{\mu/2} \leq \Big(\frac{2}{e}\Big)^{t/4}\\
    &= \Big(\frac{2}{e}\Big)^{\Theta(\sqrt{n})} = \negl(\lambda)
\end{align*}
Then by the union bound, the probability that $|I_{id}\backslash I_{adv}| \geq t/4$ for all $id\in[N-C]$ is $1-(N-C)\cdot \negl(\lambda) = 1-\negl(\lambda)$.
\\\par

\noindent 2) For convenience, let's say that $I_{id}\backslash I_{adv} = \big[|I_{id} \backslash I_{adv}|\big]$. Given that $|I_{id} \backslash I_{adv}| \geq t/4$, the following event $E$ occurs with overwhelming probability:
\[E : \quad dim\big[span(V_{adv}, \bv_1, \dots, \bv_{t/4})\big] = dim(V_{adv}) + t/4\]
\begin{align*}
    P_{\{v_i\}_{i\in[t/4]}}(E) &\geq 1 - P(\bv_1 \in V_{adv}) - \ldots 
    - P[\bv_{t/4} \in span(V_{adv}, \bv_1, \dots, \bv_{t/4 - 1})]\\
    &\geq 1 - 2^{n/4-n/2} - \ldots - 2^{n/4+t/4-1-n/2}\\
    &\geq 1 - \frac{t}{4} \cdot 2^{(t/4-n/4)} = 1 - 2^{-\Theta(n)} = 1 - \negl(\lambda)
\end{align*}

\par

\noindent 3) Putting together steps 1 and 2, we have that with overwhelming probability in $\lambda$, 
\[dim\big[span(V_{adv}, V_{id})\big] \geq dim(V_{adv}) + t/4\]
\end{proof}

\begin{lemma}
\label{thm:f}
Let $msk$ be good, let $M \in_R \mathcal{M}(A)$, and let $msk' = M(msk)$. Then for any $id \in [N-C]$ and any $\bx \in \mathbbm{Z}_2^n$,
\begin{itemize}
    \item If $\bx \not\in A$, then $P\big(\bx \in {W'_{id}}^\perp\big) = 2^{- \Omega(\sqrt{n})}$.
    \item If $\bx \not\in A^\perp$, then $P\big(\bx \in {V'_{id}}^\perp\big) = 2^{- \Omega(\sqrt{n})}$.
\end{itemize}
The probability is over the choice of $M \in_R \mathcal{M}(A)$.
\end{lemma}

\begin{proof}
We'll prove the first claim -- the second claim's proof is similar.\\

\noindent 1) Let $S = span(\{w_j\}_{j \in J_{id}\backslash J_{adv}})$. This is the random subspace that verifier $id$ has that the adversary cannot predict. We know from lemma \ref{thm:good} that $dim(S) \geq t/4$. Also $M \cdot S \leq W_{id}'$, so $W_{id}'^\perp \leq (M \cdot S)^\perp$. Then:
\begin{align*}
    P_{M}\big(\bx \in W_{id}'\big) &\leq P_{M}\big(\bx \in (M \cdot S)^\perp\big) = P_{M} \big(\bx^T \cdot M \cdot S = \mathbf{0}\big)
\end{align*}

\noindent 2) $M^T\bx$ is a random vector satisfying $M^T\bx \not\in A$. First, $M^T$ maps $A$ to $A$ and $A^\perp$ to $A^\perp$. Since $\bx \not\in A$, $\bx$ has a non-zero component in $A^\perp$, which $M^T$ maps to a non-zero component in $A^\perp$. Therefore, $M^T\bx \not\in A$.

\begin{align*}
    P_{M}\big(\bx^T \cdot M \cdot S = \mathbf{0}\big) &= P_{M}\big(M^T\bx \in S^\perp\big) \leq \frac{|S^\perp|}{|\mathbb{Z}_2^n\backslash A|}\\
    &= \frac{2^{dim(S^\perp)}}{2^n-2^{n/2}} \leq \frac{2^{n-t/4}}{2^{n-1}} = 2^{1-t/4} = 2^{-\Omega(\sqrt{n})}
\end{align*}

\end{proof}

\begin{lemma}
\label{thm:query_bound}
If $msk$ is good, then any quantum algorithm that approximately computes $f$ needs at least $2^{\Omega(\sqrt{n})}$ oracle queries.
\end{lemma}

\begin{proof}\quad\\

\noindent 1) If $O_{Full}$ and $O_{Fran}$ differ on an input, then $O_{Full}$ rejects the input, and $O_{Fran}$ accepts it. \par 

For any input $(id, s, \bx)$ to an oracle, if $O_{Full}[A](id, s, \bx) = 1$, then \\ $O_{Fran}[M(msk)](id, s, \bx) = 1$ as well. When $s = 0$, $O_{Full}$ accepts iff $\bx \in A$. Since $A \leq W_{id}^\perp$, $O_{Fran}$ accepts as well. Similar reasoning shows that when $s=1$, if $O_{Full}$ accepts, then $O_{Fran}$ accepts as well. \par 

Therefore, the only way for $O_{Full}$ and $O_{Fran}$ to give different responses to an input is if:
\[O_{Full}[A](id, s, \bx) = 0 \text{, and } O_{Fran}[M(msk)](id, s, \bx) = 1\]

\noindent 2) Lemma \ref{thm:f} says that if $O_{Full}[A](id, s, \bx) = 0$, then \[P_{M \leftarrow \mathcal{M}(A)}\Big(O_{Fran}[M(msk)](id, s, \bx) = 1\Big) = 2^{- \Omega(\sqrt{n})}\]
so we can set $u = 2^{- \Omega(\sqrt{n})}$. Also, we can set $u' = 1$ because $1$ is greater than or equal to any probability.\par

Finally, in order to approximately compute $f$, the number of oracle queries needed is $\Omega\Big(\frac{1}{\sqrt{u \cdot u'}}\Big) = 2^{\Omega(\sqrt{n})}$.
\end{proof}

\begin{lemma}
\label{thm:worst_case_adv}
For any polynomial-time quantum algorithm $\mathcal{A}$, and any good $msk$, there exists an $M \in \mathcal{M}(A)$ such that:
\[\Big|P(\mathcal{A}^{O_{Full}[A]} = 1) - P(\mathcal{A}^{O_{Fran}[M(msk)]} = 1)\Big| \leq 2^{-\Theta(\sqrt[^3]{n})}\]
\end{lemma}
\begin{proof}\quad\\
\noindent 1) Let $\Delta$ be the minimum value of \[\Big|P(\mathcal{A}^{O_{Full}[A]} = 1) - P(\mathcal{A}^{O_{Fran}[M(msk)]} = 1)\Big|\] over all $M$, and let $p = P(\mathcal{A}^{O_{Full}[A]} = 1)$.

Next, assume toward contradiction that there is some polynomial-time algorithm $\mathcal{A}$ and some good $msk$ such that $\Delta > 2^{-\Theta(\sqrt[^3]{n})}$. Then we'll construct an algorithm $\mathcal{A}'$ that approximately computes $f$ using $2^{\Theta(\sqrt[^3]{n})}$ queries (by lemma \ref{thm:query_bound}, we know this is not possible). 

$\mathcal{A}'$ runs $4n/\Delta^2$ independent iterations of $\mathcal{A}$ and averages the outputs. Let $\bar{p}$ be the average number of iterations of $\mathcal{A}$ that output $1$. Next, $\mathcal{A}'$ outputs $0$ if $|\bar{p} - p| \leq \Delta/2$ and outputs $1$ otherwise.\\

\noindent 2) $\mathcal{A}'$ gives the incorrect value for $f$ if:
\begin{enumerate}
    \item $|\bar{p} - p| \leq \Delta/2$, but the oracle is franchised.
    \item $|\bar{p} - p| > \Delta/2$, but the oracle is full.
\end{enumerate}
In the first case, $\big|\mathbbm{E}[\bar{p}] - p\big| > \Delta$, so $\big|\bar{p} - \mathbbm{E}[\bar{p}]\big| \geq \Delta/2$. In the second case as well, $\big|\bar{p} - \mathbbm{E}[\bar{p}]\big| \geq \Delta/2$.

The probability of an error is bounded by the Hoeffding inequality: 
\[P\Big(\big|\bar{p} - \mathbbm{E}[\bar{p}]\big| \geq \Delta/2\Big) \leq 2 e^{-2(\Delta/2)^2 \cdot (4n/\Delta^2)} = 2 e^{-2n}\]

Next, $\mathcal{A}'$ approximately computes $f$ because for any $O \in \mathcal{O}$, $\mathcal{A}'$ computes $f(O)$ with probability $\geq 1 - 2 e^{-2n} > 2/3$.\\

\noindent 3) Finally, $\mathcal{A}'$ makes $2^{\Theta(\sqrt[^3]{n})}$ queries. First, $\mathcal{A}$ makes $2^{O(\log n)}$ queries because it runs in polynomial time. So the number of queries that $\mathcal{A}'$ makes is:
\[\frac{4n}{\Delta^2} \cdot 2^{O(\log n)} = 2^{O(\log n) + O(\sqrt[^3]{n})} = 2^{O(\sqrt[^3]{n})}\]

Since no algorithm can approximately compute $f$ using $2^{O(\sqrt[^3]{n})}$ queries, this is a contradiction. So for any polynomial-time $\mathcal{A}$, and any good $msk$, there exists an $M$ such that \[\Big|P(\mathcal{A}^{O_{Full}[A]} = 1) - P(\mathcal{A}^{O_{Fran}[M(msk)]} = 1)\Big| \leq 2^{-\Theta(\sqrt[^3]{n})}\]
\end{proof}

\begin{lemma}
\label{thm:avg_case_adv}
For any polynomial-time quantum algorithm $\mathcal{A}$, any good $msk$, and a uniformly random $M \in_R \mathcal{M}(A)$, 
\[\Big|P(\mathcal{A}^{O_{Full}[A]} = 1) - P(\mathcal{A}^{O_{Fran}[M(msk)]} = 1)\Big| \leq 2^{-\Theta(\sqrt[^3]{n})}\]
The probability is over $\mathcal{A}$'s randomness and the choice of $M$.
\end{lemma}
Note that lemma \ref{thm:avg_case_adv} is equivalent to theorem \ref{thm:dist_adv}.

\begin{proof}
The problem of distinguishing full and franchised oracles is random self-reducible. Since lemma \ref{thm:worst_case_adv} says the algorithm's distinguishing advantage is negligible in the worst case, then their advantage is also negligible in the average case.

Assume toward contradiction that there exists a polynomial-time quantum algorithm $\mathcal{A}$ such that for a uniformly random  $M \in_R \mathcal{M}(A)$, 
\[\delta := \Big|P(\mathcal{A}^{O_{Full}[A]} = 1) - P(\mathcal{A}^{O_{Fran}[M(msk)]} = 1)\Big| = 2^{-o(\sqrt[^3]{n})}\] 
Then we'll construct a polynomial-time algorithm $\mathcal{A}'$ that runs $\mathcal{A}$ as a subroutine and achieves $\delta = 2^{-o(\sqrt[^3]{n})}$ for all $M$ (by lemma \ref{thm:worst_case_adv}, this is impossible).

Given any $M \in \mathcal{M}(A)$, $\mathcal{A}'$ samples a uniformly random $R \in_R \mathcal{M}(A)$. Then $R[M(msk)]$ is an ``average-case'' master secret key in the sense that $R[M(msk)] = (R \cdot M)(msk)$, and $R' := R \cdot M$ is uniformly random in $\mathcal{M}(A)$.

$\mathcal{A}'$ gives $msk$ to $\mathcal{A}$ and simulates the distinguishing game in which the franchised verifiers are using $R[M(msk)]$. Whenever $\mathcal{A}$ queries the oracle, $\mathcal{A}'$ uses $R$ as a change-of-basis for the query before forwarding it to the challenger. In $\mathcal{A}$'s view, it is dealing with a uniformly random $R' \in \mathcal{M}(A)$, so $\mathcal{A}$ has distinguishing advantage $\delta$. Therefore, $\mathcal{A}'$ has the same advantage $\delta = 2^{-o(\sqrt[^3]{n})}$, but for every $M$. This contradicts lemma \ref{thm:worst_case_adv}, so in fact, lemma \ref{thm:avg_case_adv}'s claim is true.
\end{proof}

Lemma \ref{thm:avg_case_adv} proves theorem \ref{thm:dist_adv}.